\documentclass[12pt]{article}
\usepackage{amsfonts}
\usepackage{amsmath}
\usepackage[onehalfspacing]{setspace}
\usepackage{bbm}

\newtheorem{theorem}{Theorem}

\newtheorem{lemma}[theorem]{Lemma}

\newtheorem{proposition}[theorem]{Proposition}

\newenvironment{proof}[1][Proof]{\noindent\textbf{#1.} }{\ \rule{0.5em}{0.5em}}
\input{tcilatex}
\begin{document}

\title{The Better Half of Selling Separately\thanks{%
Previous versions: December 2017 (\texttt{https://arxiv.org/abs/1712.08973}%
), November 2016. We thank Noam Nisan for useful discussions, and the
editors and referees for their helpful suggestions and comments.}}
\author{Sergiu Hart\thanks{%
Department of Economics, Institute of Mathematics, and Federmann Center for
the Study of Rationality, The Hebrew University of Jerusalem.\quad \emph{%
E-mail}: \texttt{hart@huji.ac.il} \quad \emph{Web site}: \texttt{%
http://www.ma.huji.ac.il/hart}} \and Philip J. Reny\thanks{%
Department of Economics, University of Chicago. Research partially supported
by National Science Foundation grant SES-1724747.\quad \emph{E-mail}: 
\texttt{preny@uchicago.edu} \quad \emph{Web site}: \texttt{%
http://economics.uchicago.edu/directory/philip-j-reny}}}
\maketitle

\begin{abstract}
Separate selling of two independent goods is shown to yield at least $62\%$
of the optimal revenue, and at least $73\%$ when the goods satisfy the
Myerson regularity condition. This improves the $50\%$ result of Hart and
Nisan (2017\textbf{, }originally circulated in 2012).
\end{abstract}

\tableofcontents

\def\@biblabel#1{#1\hfill}
\def\thebibliography#1{\section*{References}
\addcontentsline{toc}{section}{References}
\list
{}{
\labelwidth 0pt
\leftmargin 1.8em
\itemindent -1.8em
\usecounter{enumi}}
\def\newblock{\hskip .11em plus .33em minus .07em}
\sloppy\clubpenalty4000\widowpenalty4000
\sfcode`\.=1000\relax\def\baselinestretch{1}\large \normalsize}
\let\endthebibliography=\endlist%

\section{Introduction\label{s:introduction}}

One of the most celebrated aspects of Myerson's (1981) optimal auction
result is that it provides an economic explanation for the ubiquitous use of
the four standard auction forms. Strictly speaking, however, Myerson's
results apply only to cases in which a seller is selling a single good.
Because many sellers sell multiple goods, extending Myerson's analysis to
the multi-good case has long been considered a critical next step. But the
multi-good monopoly problem has resisted a complete solution for over 35
years and by now it is well understood to be an extremely difficult problem.
Worse still, it is known that the optimal solution must typically be quite
complex and very often requires buyers to purchase randomized contracts. And
therein lies the difficulty, because we do not often, if ever, observe
complex or randomized selling mechanisms in practice. This raises the
obvious question, Why not?

One reason that we may not observe the kinds of complex mechanisms that full
optimality dictates is that relatively simple mechanisms may suffice for
generating much\ of the revenue that could ever be generated. Thus, a
complementary approach to the research program initiated in Myerson (1981)
is to search for relatively simple mechanisms that yield a significant
fraction of the revenue that is generated by a fully optimal mechanism.%
\footnote{%
For examples of this approach, see Hart and Nisan (2017) and the references
therein.} The present paper represents a modest contribution to this program.

We consider a single seller who has one unit each of two indivisible goods.
The two goods need not be identical. The seller, whose value for the two
goods is zero, can offer to sell the goods to a single buyer. The buyer's
two values, one value for each of the two goods, are unknown to the seller
but are known to be independently distributed (so we say that the two goods
are independent). The buyer is risk neutral and has preferences that are
additive in the values and (negatively so) in the price paid. Even in this
most basic case, there is no known characterization of the optimal selling
mechanism, though it is known that the optimal mechanism can display unusual
properties.\footnote{%
For example, the optimal mechanism can be non-monotonic (i.e., increasing
the buyer's valuations may well decrease the seller's optimal revenue), and
it can require the buyer to accept randomized contracts. See Hart and Reny
(2015); Section 3 there contains references to previous examples that
require such randomizations.} We ask, What fraction of the optimal revenue
can the seller guarantee by selling each of the two goods separately, i.e.,
by posting a Myerson-optimal price for each of the goods?

In the context of a general analysis with any finite number of independent
goods, Hart and Nisan (2017\textbf{, }originally circulated in 2012) show in
particular that, by selling two independent goods separately, the seller can
guarantee at least $50\%$ of the optimal revenue but cannot guarantee more
than\footnote{%
Hart and Nisan (2017) establish the $78\%$ upper bound with an explicit
example in which it is optimal to sell two independent and identically
distributed goods as a bundle (these goods satisfy the Myerson-regularity
condition).} $78\%$. A nice feature of the $50\%$ revenue guarantee is that
its proof is relatively simple. In part, this simplicity arises from the
rather generous bounds that are established at various steps. While it seems
clear that the bounds employed in the Hart--Nisan proof are
\textquotedblleft much\textquotedblright\ too generous, tightening them as
we do here requires a surprising amount of additional effort. Hart and Nisan
also show that if, in addition, the buyer's two independent values are \emph{%
identically} distributed, then the revenue guarantee is at least $73\%$,
which is tantalizingly close to the $78\%$ upper limit.

Our main result significantly improves upon the Hart--Nisan $50\%$
guarantee, and shows that their $73\%$ guarantee with identically
distributed values can also be obtained when the buyer's value distributions
satisfy Myerson-regularity. None of our results require the two values to be
identically distributed.

\begin{quote}
\noindent \textbf{Main Result.}\textsc{\ }\emph{For any two independent
goods, selling each good separately at its optimal one-good price guarantees
at least}\linebreak \emph{\ }$\sqrt{e}/(\sqrt{e}+1)\approx 62\%$\emph{\ of
the optimal revenue. Furthermore, if the buyer's two value distributions
each satisfy Myerson-regularity, then the guaranteed fraction of optimal
revenue increases to\linebreak\ }$e/(e+1)\approx 73\%.$
\end{quote}

\noindent This is stated below as Theorems \ref{th:62.2} and \ref{th:regular}%
.

To summarize the known bounds on the guaranteed fraction of optimal revenue (%
\textsc{gfor}) from selling separately two goods: when the goods are
independent, the \textsc{gfor} is at least $62\%;$ when they are independent
and either Myerson-regular or identically distributed, the \textsc{gfor} is
at least $73\%;$ in all these three cases, the \textsc{gfor} is at most $%
78\%;$ and, when the goods are not necessarily independent, the \textsc{gfor}
drops all the way down to zero (Hart and Nisan 2013).\footnote{%
For more than two goods a similar result is due to Briest, Chawla,
Kleinberg, and Weinberg (2015, originally circulated in 2010).}

\subsection{Some Related Work}

There is by now a vast literature in game theory, economics, and computer
science that deals with the (optimal) selling of multiple goods. While that
literature is too large to survey here, the reader may wish to consult the
literature section in, say, Hart and Nisan (2017) for an overview. We will
mention here the work of Babaioff, Immorlica, Lucier, and Weinberg (2014)
that shows that the better option between selling the goods separately and
selling them as the bundle of all goods yields a \textsc{gfor} that is
bounded away from zero \emph{for any number of goods.} Recently, in the case
of two independent goods, Babaioff, Nisan, and Rubinstein (2018) have shown
that separate selling yields at least $78\%$ of the optimal \emph{%
deterministic} revenue, and that this bound is tight. In the related setup
of a \emph{unit-demand buyer} (who desires to buy only one good, rather than
having an additive value over bundles of goods), Chawla, Malec, and Sivan
(2010, Theorem 5) show a \textsc{gfor} of $1/4$ for the separate selling of
any number of independent goods. Finally, Daskalakis, Deckelbaum, and Tzamos
(2017) provide a useful duality characterization of the revenue-optimizing
mechanism for multiple goods.

\subsection{Organization of the Paper}

The paper is organized as follows. Section \ref{s:preliminaries} presents
the model, defines the appropriate concepts, and provides some preliminary
results. Section \ref{sus:overview} gives an outline of the proof. The proof
itself consists of a general decomposition result (Proposition \ref%
{p:revenue} in Section \ref{s:bound}) and an estimate of the crucial term
there (Proposition \ref{p:bound-I} in Section \ref{s:K1-K2}), which, when
combined, give the first part of the Main Theorem, namely, the general $62\%$
bound (Theorem \ref{th:62.2} in Section \ref{s:proof of thm A}). Section \ref%
{s:73} proves the second part of the Main Theorem, namely, the $73\%$ bound
for regular distributions (Theorem \ref{th:regular}), together with some
additional results. Appendix \ref{s:continuity} provides a general result on
the continuity of the revenue with respect to valuations (which is of
independent interest), and Appendix \ref{s:nonsymm} gives a simple
illustration of the use of nonsymmetric diagonals.

\section{Preliminaries\label{s:preliminaries}}

\subsection{The Model\label{sus:model}}

The basic model is standard, and the notation follows Hart and Reny (2015)
and Hart and Nisan (2017), which the reader may consult for further details
and references.

One seller (or \textquotedblleft monopolist") is selling a number $k\geq 1$
of goods\ (or \textquotedblleft items," \textquotedblleft objects," etc.) to
one buyer.

The goods have no value or cost to the seller. Let $x_{1},x_{2},...,x_{k}%
\geq 0$ be the buyer's values for the goods. The value for getting a set of
goods is \emph{additive}: getting the subset $I\subseteq \{1,2,...,k\}$ of
goods is worth $\sum_{i\in I}x_{i}$ to the buyer (and so, in particular, the
buyer's demand is \emph{not} restricted to one good only). The valuation of
the goods is given by a random variable $X=(X_{1},X_{2},...,X_{k})$ that
takes values in $\mathbb{R}_{+}^{k}$ (we thus assume that valuations are
always nonnegative); we will refer to $X$ as a $k$\emph{-good} \emph{random
valuation.} The realization $x=(x_{1},x_{2},...,x_{k})\in \mathbb{R}_{+}^{k}$
of $X$ is known to the buyer, but not to the seller, who knows only the
distribution $F$ of $X$ (which may be viewed as the seller's belief); we
refer to a buyer with valuation $x$ also as a buyer of \emph{type }$x$. The
buyer and the seller are assumed to be risk neutral and to have quasi-linear
utilities.

The objective is to \emph{maximize} the seller's (expected) \emph{revenue}.

By the Revelation Principle\ (Myerson 1981), it is without loss of
generality to restrict attention to \textquotedblleft direct
mechanisms\textquotedblright\ that are \textquotedblleft incentive
compatible.\textquotedblright\ A direct\emph{\ mechanism} $\mu $ consists of
a pair of functions\footnote{%
All functions in this paper are assumed to be Borel measurable (cf. Hart and
Reny 2015, footnotes 10 and 48).} $(q,s),$ where $q=(q_{1},q_{2},...,q_{k}):%
\mathbb{R}_{+}^{k}\rightarrow \lbrack 0,1]^{k}$ and\footnote{%
Without loss of generality any mechanism can always be extended to the whole
space $\mathbb{R}_{+}^{k};$ see Hart and Reny (2015).} $s:\mathbb{R}%
_{+}^{k}\rightarrow \mathbb{R}.$ If the buyer reports a valuation vector $%
x\in \mathbb{R}_{+}^{k},$ then $q_{i}(x)\in \lbrack 0,1]$ is the probability
that the buyer receives good\footnote{%
When the goods are infinitely divisible and the valuations are linear in
quantities, $q_{i}$ may be alternatively viewed as the \emph{quantity} of
good $i$ that the buyer gets.} $i$ (for $i=1,2,...,k$), and $s(x)$ is the
payment that the seller receives from the buyer\emph{.} When the buyer
reports his value $x$ truthfully, his payoff is\footnote{%
The scalar product of two $n$-dimensional vectors $y=(y_{1},...,y_{n})$ and $%
z=(z_{1},...,z_{n})$ is $y\cdot z=\sum_{i=1}^{n}y_{i}z_{i}$.} $%
b(x)=\sum_{i=1}^{k}q_{i}(x)x_{i}-s(x)=q(x)\cdot x-s(x),$ and the seller's
payoff is $s(x).$

The mechanism $\mu =(q,s)$ satisfies \emph{individual rationality} (\textbf{%
IR)} if $b(x)\geq 0$ for every\footnote{%
Individual rationality recognizes that, regardless of his valuation, the
buyer can obtain an expected payoff of zero by not participating in the
mechanism.} $x\in \mathbb{R}_{+}^{k};$ it satisfies \emph{incentive
compatibility} (\textbf{IC}) if $b(x)\geq q(\tilde{x})\cdot x-s(\tilde{x})$
for every alternative report $\tilde{x}\in \mathbb{R}_{+}^{k}$ of the buyer
when his value is $x,$ for every $x\in \mathbb{R}_{+}^{k};$ and it satisfies 
\emph{no positive transfer} (\textbf{NPT)} if $s(x)\geq 0$ for every $x\in 
\mathbb{R}_{+}^{k}$ (which, together with IR, implies that $s(0)=b(0)=0).$

The (expected) revenue of a mechanism $\mu =(q,s)$ from a buyer with random
valuation $X,$ which we denote by $R(\mu ;X),$ is the expectation of the
payment received by the seller; i.e., $R(\mu ;X)=\mathbb{E}\left[ s(X)\right]
.$ We now define:

\begin{itemize}
\item \textsc{Rev}$(X),$ the \emph{optimal revenue}, is the maximal revenue
that can be obtained: \textsc{Rev}$(X):=\sup_{\mu }R(\mu ;X),$ where the
supremum is taken over all mechanisms $\mu $ that satisfy IR and IC.
\end{itemize}

When there is only one good, i.e., when $k=1,$ Myerson's (1981) result is
that 
\begin{equation}
\text{\textsc{Rev}}(X)=\sup_{p\geq 0}p\cdot \mathbb{P}\left[ X\geq p\right]
=\sup_{p\geq 0}p\cdot \mathbb{P}\left[ X>p\right] =\sup_{p\geq 0}p\cdot
(1-F(p)),  \label{eq:one good}
\end{equation}%
where $F$ is the cumulative distribution function of $X.$ Optimal mechanisms
correspond to the seller \textquotedblleft posting" a price $p$ and the
buyer buying the good for the price $p$ whenever his value is at least $p$;
in other words, the seller makes the buyer a \textquotedblleft
take-it-or-leave-it" offer to buy the good at price $p.$

Besides the maximal revenue \textsc{Rev}$(X)$, we consider what can be
obtained from the simple class of mechanisms that sell each good separately.

\begin{itemize}
\item \textsc{SRev}$(X),$ the \emph{separate revenue}, is the maximal
revenue that can be obtained by selling each good separately. Thus 
\begin{equation*}
\text{\textsc{SRev}}(X)=\text{\textsc{Rev}}(X_{1})+\text{\textsc{Rev}}%
(X_{2})+...+\text{\textsc{Rev}}(X_{k}).
\end{equation*}
\end{itemize}

\noindent The separate revenue is obtained by solving $k$ one-dimensional
problems (using (\ref{eq:one good})), one for each good.

We now state the basic properties from Hart and Nisan (2017, Propositions 5
and 6) needed for our proof.

\begin{proposition}
\label{p:HN-p-5-6}

\begin{description}
\item[(i)] Let $\mu =(q,s)$ be a mechanism for $k$ goods with buyer payoff
function $b.$ Then $\mu =(q,s)$ satisfies IC if and only if $b$ is a convex
function and for all $x$ the vector $q(x)$ is a subgradient of $b$ at $x$
(i.e., $b(\tilde{x})-b(x)\geq q(x)\cdot (\tilde{x}-x)$ for all $\tilde{x}$).

\item[(ii)] \textsc{Rev}$(X)=\sup_{\mu }R(\mu ;X)$ with the supremum taken
over all IC, IR, and NPT mechanisms $\mu $.
\end{description}
\end{proposition}

\subsection{Distributions\label{sus:continuous distributions}}

As we show formally in Appendix \ref{sus:wlog-continuity}, for the results
of the present paper we can limit ourselves without loss of generality to
valuations that admit a density function (this follows from general
continuity properties of the revenue, which we prove in Appendix \ref%
{s:continuity}, and are of independent interest).

In what follows we thus assume that every nonnegative random variable $X$
has an absolutely continuous cumulative distribution function, $F(t)=\mathbb{%
P}\left[ X\leq t\right] =\mathbb{P}\left[ X<t\right] ,$ with an associated
density function $f(t)$. We denote by $G$ the tail probability, i.e.,%
\begin{equation*}
G(t)%
{\;:=\;}%
1-F(t)=\int_{t}^{\infty }f(u)\mathrm{d}u=\mathbb{P}\left[ X\geq t\right] ,
\end{equation*}%
and by $H$ the cumulative tail probability, i.e.,%
\begin{equation}
H(t)%
{\;:=\;}%
\int_{0}^{t}G(u)\mathrm{d}u=\mathbb{E}\left[ \min \{X,t\}\right]
\label{eq:H}
\end{equation}%
(the equality holds since $\mathbb{E}\left[ \min \{X,t\}\right]
=\int_{0}^{\infty }\mathbb{P}\left[ \min \{X,t\}\geq u\right] \mathrm{d}%
u=\allowbreak $\linebreak $\int_{0}^{t}\mathbb{P}\left[ X\geq u\right] 
\mathrm{d}u=\int_{0}^{t}G(u)\mathrm{d}u).$

Let $r:=$\textsc{Rev}$(X)>0$ be\footnote{%
The continuity of $F$ implies that $X$ cannot be identically zero, and so
the optimal revenue \textsc{Rev}$(X)$ must be positive (just sell the good
at a small enough positive price).} the optimal revenue from $X;$ then (\ref%
{eq:one good}) implies $G(t)\leq r/t,$ which together with $G(t)\leq 1$ gives%
\begin{equation*}
G(t)\leq \min \left\{ \frac{r}{t},1\right\} .
\end{equation*}%
Therefore%
\begin{equation}
H(t)\leq \int_{0}^{r}1\mathrm{d}u+\int_{r}^{t}\frac{r}{u}\mathrm{d}u=r+r\log 
\frac{t}{r},  \label{eq:bound-H-0}
\end{equation}%
for every $t\geq r$ (and $H(t)\leq t$ for $t\leq r)$.

\subsection{Change of Units\label{sus:lambdas}}

We start with a trivial, but useful, change of units. For every $0<\lambda
_{1},...,\lambda _{k}\leq 1,$ let $\mathcal{M}_{\lambda _{1},...,\lambda
_{k}}$ denote the set of all IC and IR mechanisms $\mu =(q,s)$ that satisfy $%
q_{i}(x)\in \lbrack 0,\lambda _{i}]$ (instead of $q_{i}(x)\in \lbrack 0,1])$
for every $x\in \mathbb{R}_{+}^{k}$ and $i=1,...,k.$ The set of all IC and
IR mechanisms, which we denote by $\mathcal{M},$ is thus the same as $%
\mathcal{M}_{1,...,1}.$

\begin{lemma}
\label{l:rescale-rev}For every $0<\lambda _{1},...,\lambda _{k}\leq 1$ we
have%
\begin{equation*}
\text{\textsc{Rev}}(X_{1},...,X_{k})=\sup_{\mu \in \mathcal{M}_{\lambda
_{1},...,\lambda _{k}}}R(\mu ;\tilde{X}_{1},...,\tilde{X}_{k}),
\end{equation*}%
where $\tilde{X}_{i}:=(1/\lambda _{i})X_{i}$ for $i=1,...,k.$
\end{lemma}

\begin{proof}
Given $\mu =(q,s)$ with $q_{i}(x)\in \lbrack 0,\lambda _{i}]$ for all $i,$
define $\hat{\mu}=(\hat{q},\hat{s})$ by $\hat{q}_{i}(x_{1},...,x_{k}):=(1/%
\lambda _{i})q_{i}(x_{1}/\lambda _{1},...,x_{k}/\lambda _{k})\in \lbrack
0,1] $ and $\hat{s}(x_{1},...,x_{k}):=s(x_{1}/\lambda _{1},...,x_{k}/\lambda
_{k}) $ (and thus $\hat{b}(x_{1},...,x_{k})=b(x_{1}/\lambda
_{1},...,x_{k}/\lambda _{k})$ for the corresponding buyer's payoff
functions). It is immediate to see that $\hat{\mu}$ is IC and IR if and only
if $\mu $ is IC and IR, and that $\mathbb{E}\left[ s(\tilde{X}_{1},...,%
\tilde{X}_{k})\right] =\mathbb{E}\left[ \hat{s}(X_{1},...,X_{k})\right] .$
Conversely, given $\hat{\mu}$ one generates $\mu $ by the reverse
transformation.
\end{proof}

\section{Overview of the Proof\label{sus:overview}}

The first part of the proof is similar to the proofs of Theorems A and B in
Hart and Nisan (2017)\footnote{%
The reader is encouraged to look at these proofs and the explanations there.}
except that, where they split the buyer's space of values $(x_{1},x_{2})\in 
\mathbb{R}_{+}^{2}$ in half along the diagonal $x_{1}=x_{2}$, we split the
space into two regions $x_{1}\geq \lambda x_{2}$ and $x_{1}<\lambda x_{2}$
along a possibly nonsymmetric diagonal $x_{1}=\lambda x_{2}$ (the precise
value of $\lambda $ will be chosen later). For any two-good mechanism, the
revenue in each of the two regions can be estimated by constructing from it
appropriate one-good mechanisms, which eventually leads to a key bound: see
Proposition \ref{p:revenue} in Section \ref{s:bound}. (Rather than working
directly with the two asymmetric regions, which is cumbersome, the proof
simplifies computations by first making an appropriate change of units,
which amounts to rescaling the probabilities that the goods are received:
see Lemma \ref{l:rescale-rev} in Section \ref{sus:lambdas}.) Once we have
the bound given in Proposition \ref{p:revenue}, we need to estimate the
maximum of a certain integral expression---which is essentially the
additional revenue that is achievable beyond the separate one-good revenues%
\textbf{---}over pairs of nonnegative functions $\varphi _{1},\varphi _{2}$
whose sum $\varphi _{1}+\varphi _{2}$ is nondecreasing. This is accomplished
in Proposition \ref{p:I-lambda}\textbf{, }by considering the appropriate
extreme functions and then carefully estimating the relevant terms (this is
the hardest part of the proof). In Section \ref{s:proof of thm A} we put
everything together, and, by choosing the best possible $\lambda $
(specifically,\textbf{\ }$\lambda =1/\sqrt{e}),$ prove the $62\%$ bound
(Theorem \ref{th:62.2}). Then in Section \ref{s:73} we show the $73\%$ bound
for Myerson-regular goods (Theorem \ref{th:regular}), and then we also deal
with monotonic mechanisms. There are two appendices: Appendix \ref%
{s:continuity} establishes that, under quite permissive conditions, the
seller's revenue is continuous in the distribution of the buyer's valuation,
a result that we use in our proof, but that is also of independent interest,
and Appendix \ref{s:nonsymm} provides a simple illustration of how the
\textquotedblleft nonsymmetric diagonal" construct alone can produce useful
bounds.

\section{Bounding the Revenue by Nonsymmetric Decomposition\label{s:bound}}

This section provides the basic decomposition\textbf{\ }with respect to a%
\textbf{\ }\emph{nonsymmetric}\textbf{\ }diagonal (equivalently, we make a
corresponding change of units and use the symmetric diagonal; see Section %
\ref{sus:lambdas}).\textbf{\ }

Given a two-good random valuation $(X_{1},X_{2}),$ for $i=1,2$ let $F_{i}$
denote the cumulative distribution function of $X_{i},$ and let $%
f_{i},G_{i}, $ and $H_{i}$ be the associated funtions as defined in Section %
\ref{sus:continuous distributions} (namely, the density, tail probability,
and cumulative tail probability functions, respectively). We let $r_{i}:=$%
\textsc{Rev}$(X_{i})$ be the optimal revenue that can be obtained from good $%
i,$ and define two useful auxiliary functions $K_{1}$ and $K_{2}$: 
\begin{eqnarray}
&&K_{1}(t)%
{\;:=\;}%
f_{2}(t)(H_{1}(t)-r_{1})-G_{1}(t)G_{2}(t),  \label{eq:K1} \\
&&K_{2}(t)%
{\;:=\;}%
f_{1}(t)(H_{2}(t)-r_{2})-G_{1}(t)G_{2}(t).  \label{eq:K2}
\end{eqnarray}

The following lemma, which slightly generalizes Lemma 19 in Hart and Nisan
(2017) (it replaces the factor $1-q(x_{0})$\ there with $\lambda -q(x_{0})$
here), obtains a better bound on the revenue of a mechanism by
\textquotedblleft rescaling" its allocation function $q$ so that it covers
the entire interval $[0,\lambda ].$

\begin{lemma}
\label{l:constrained-R1}Let $X$ be a one-good random valuation with values
bounded from below by some $x_{0}\geq 0.$ Then for every IC mechanism $\mu
=(q,s)$ that satisfies $q(x)\leq \lambda $ for all\footnote{%
It suffices to require $q(x)\leq \lambda $ for $x$ in the support of $X.$ As
in Hart and Reny (2015), one can always extend a $k$-good mechanism to the
whole space $\mathbb{R}_{+}^{k}$ without increasing its menu beyond taking
closure, and so the bound extends to all $\mathbb{R}_{+}^{k}.$} $x\geq x_{0}$
we have 
\begin{equation}
R(\mu ;X)\leq (\lambda -q(x_{0}))\text{\textsc{Rev}}(X)+s(x_{0}).
\label{eq:q0-b0}
\end{equation}
\end{lemma}

\begin{proof}
The function $q$ is nondecreasing (because $q$ is the derivative of the
buyer's payoff function $b,$ which is convex), and so $q(x_{0})\leq q(x)\leq
\lambda $ for all\footnote{%
If the values of $X$ are bounded from above by some finite $x_{1}$, then we
can replace $\lambda $ with $q(x_{1}).$} $x\geq x_{0}.$

If $q(x_{0})=\lambda $ then $q(x)=q(x_{0})=\lambda $ for all $x\geq x_{0},$
hence $s(x)=s(x_{0})$ for all $x\geq x_{0}$ by IC; therefore $\mathbb{E}%
\left[ s(X)\right] =s(x_{0})$ and (\ref{eq:q0-b0}) holds as equality.

If $q(x_{0})<\lambda $ then define a new mechanism $\hat{\mu}=(\hat{q},\hat{s%
})$ by $\hat{q}(x):=\theta (q(x)-q(x_{0}))$ and $\hat{s}(x):=\theta
(s(x)-s(x_{0})),$ and thus $\hat{b}(x):=\theta
(b(x)-(x-x_{0})q(x_{0})-b(x_{0})),$ where $\theta :=1/(\lambda -q(x_{0}))>0$
(so that $0\leq \hat{q}(x)\leq 1)$. It is immediate to verify that $(\hat{q},%
\hat{s})$ is an IC and IR mechanism: indeed, $[\hat{q}(x)\cdot x-\hat{s}%
(x)]-[\hat{q}(x^{\prime })\cdot x-\hat{s}(x^{\prime })]=\theta \left(
\lbrack q(x)\cdot x-s(x)]-[q(x^{\prime })\cdot x-s(x^{\prime })]\right) \geq
0,$ and $\hat{b}(x_{0})=0.$ Therefore \textsc{Rev}$(X)\geq \mathbb{E}\left[ 
\hat{s}(X)\right] =\theta (\mathbb{E}\left[ s(X)\right] -s(x_{0})),$ which
yields (\ref{eq:q0-b0}).
\end{proof}

\bigskip

We now come to the main result of this section, which generalizes the
decomposition of the proofs of Theorems A and B in Hart and Nisan (2017)%
\textbf{: }the revenue from two goods is bounded by the sum of the separate
one-good revenues and an additional term (the $K_{i}$-term), which will be
estimated in the next section.

\begin{proposition}
\label{p:revenue}Let $X=(X_{1},X_{2})$ be a two-good random valuation with
independent goods (i.e., $X_{1}$ and $X_{2}$ are independent nonnegative
real random variables), and let $\mu =(q,s)$ be a two-good IC, IR, and NPT
mechanism that satisfies $q_{i}(x)\leq \lambda _{i}$ for all $x\in \mathbb{R}%
_{+}^{2}$ and $i=1,2.$ Then there exist functions $\varphi _{i}:\mathbb{R}%
_{+}\rightarrow \lbrack 0,\lambda _{i}]$ for $i=1,2$ such that $\varphi
_{1}+\varphi _{2}$ is a nondecreasing function and 
\begin{equation}
R(\mu ;X_{1},X_{2})\leq \lambda _{1}r_{1}+\lambda _{2}r_{2}+\int_{0}^{\infty
}(\varphi _{1}(t)K_{1}(t)+\varphi _{2}(t)K_{2}(t))\mathrm{d}t;
\label{eq:revenue}
\end{equation}%
specifically, $\varphi _{i}(t)=q_{i}(t,t).$
\end{proposition}

\begin{proof}
The first part of the proof, which yields (\ref{eq:E[s]}), goes along the
same lines as the proof of Theorem B in Appendix A.1 of Hart and Nisan
(2017), but with the appropriate modifications, because here $X_{1}$ and $%
X_{2}$ are not identically distributed, the mechanism $\mu $ is not
symmetric, and each $q_{i}$ is bounded by $\lambda _{i}.$

We will write $Y$ for $X_{1}$ and $Z$ for $X_{2},$ and so $X=(Y,Z).$

For every $t\geq 0$ define\footnote{%
Notice that $\Phi $ here is $2\Phi $ in Hart and Nisan (2017).} $\Phi
(t):=b(t,t)$ and $\varphi _{i}(t):=q_{i}(t,t).$ By Proposition \ref%
{p:HN-p-5-6}(i) the function $\Phi $ is convex and $q(t,t)=(\varphi
_{1}(t),\varphi _{2}(t))$ is a subgradient of $b$ at $(t,t)$, and so $%
\varphi _{1}(t)+\varphi _{2}(t)$ is a subgradient of $\Phi $ at $t.$
Therefore $\varphi _{1}+\varphi _{2}$ is a nondecreasing function, and $\Phi
(u)=\int_{0}^{u}(\varphi _{1}(t)+\varphi _{2}(t))$\textrm{d}$t$ (use
Corollary 24.2.1 in Rockafellar 1970, recalling that $\Phi (0)=b(0,0)=0$ by
NPT).

Consider first the region $Y\geq Z.$ For each fixed value $z\geq 0$ of the
second good such that $\mathbb{P}\left[ Y\geq z\right] >0,$ define a
mechanism $\mu ^{z}=(q^{z},s^{z})$\ for the first good by replacing the
allocation of the second good with an equivalent decrease in payment; that
is, the allocation of the first good is unchanged, i.e., $%
q^{z}(y):=q_{1}(y,z),$\ and the payment is $s^{z}(y):=s(y,z)-zq_{2}(y,z),$\
for every $y\geq 0;$\ note that the buyer's payoff remains the same: $%
b^{z}(y)=b(y,z)$. The mechanism $\mu ^{z}$ is IC and IR for $y$, since $\mu $
is IC and IR for $(y,z)$. Let $Y^{z}$ denote the random variable $Y$
conditional on the event $Y\geq z,$ and consider the revenue $R(\mu
^{z};Y^{z})=\mathbb{E}\left[ s^{z}(Y^{z})\right] =\mathbb{E}\left[
s^{z}(Y)|Y\geq z\right] $ of $\mu ^{z}$ from $Y^{z}$. We have $Y^{z}\geq z,$ 
$q^{z}(z)=q_{1}(z,z)=\varphi _{1}(z),$ and $%
s^{z}(z)=s(z,z)-zq_{2}(z,z)=zq_{1}(z,z)-b(z,z)=z\varphi _{1}(z)-\Phi (z),$
and so, applying Lemma \ref{l:constrained-R1} above to $Y^{z},$ we have 
\begin{equation}
\mathbb{E}\left[ s^{z}(Y)|Y\geq z\right] \leq (\lambda _{1}-\varphi _{1}(z))%
\text{\textsc{Rev}}(Y^{z})+z\varphi _{1}(z)-\Phi (z).  \label{eq:X-tilda}
\end{equation}%
Since $\mathbb{P}\left[ Y_{z}\geq t\right] =\mathbb{P}\left[ Y\geq t\right] /%
\mathbb{P}\left[ Y\geq z\right] =G_{1}(t)/\mathbb{P}\left[ Y\geq z\right] $
for all $t\geq z,$ using (\ref{eq:one good}) we get%
\begin{equation*}
\text{\textsc{Rev}}(Y^{z})=\sup_{t\geq 0}t\cdot \mathbb{P}\left[ Y^{z}\geq t%
\right] =\sup_{t\geq z}t\cdot \frac{G_{1}(t)}{\mathbb{P}\left[ Y\geq z\right]
}\leq \frac{\sup_{t\geq 0}t\cdot G_{1}(t)}{\mathbb{P}\left[ Y\geq z\right] }=%
\frac{r_{1}}{\mathbb{P}\left[ Y\geq z\right] }
\end{equation*}%
(recall that $r_{1}=$\textsc{Rev}$(Y)).$ Substitute this into (\ref%
{eq:X-tilda}), and multiply it by $\mathbb{P}\left[ Y\geq z\right] ,$ to get%
\begin{equation*}
\mathbb{E}\left[ s^{z}(Y)\mathbf{1}_{Y\geq z}\right] \leq (\lambda
_{1}-\varphi _{1}(z))r_{1}+(z\varphi _{1}(z)-\Phi (z))\mathbb{P}\left[ Y\geq
z\right]
\end{equation*}%
for all $z\geq 0$ (which trivially includes those $z$ where $\mathbb{P}\left[
Y\geq z\right] =0).$ Taking expectation over the values $z$ of $Z$ yields%
\begin{equation}
\mathbb{E}\left[ s^{Z}(Y)\mathbf{1}_{Y\geq Z}\right] \leq \lambda
_{1}r_{1}-r_{1}~\mathbb{E}\left[ \varphi _{1}(Z)\right] +\mathbb{E}\left[
(Z\varphi _{1}(Z)-\Phi (Z)\mathbf{1}_{Y\geq Z}\right] .  \label{eq:e[s-hat]}
\end{equation}

For $y\geq z\geq 0$ we have $s(y,z)=s^{z}(y)+zq_{2}(y,z)\,\leq
s^{z}(y)+zq_{2}(y,y)=s^{z}(y)+z\varphi _{2}(y)$ (by the monotonicity of $%
q_{2}$ in its second variable, again by the convexity of $b),$ which
together with (\ref{eq:e[s-hat]}) yields%
\begin{eqnarray*}
\mathbb{E}\left[ s(Y,Z)\mathbf{1}_{Y\geq Z}\right] &\leq &\mathbb{E}\left[
s^{Z}(Y)\mathbf{1}_{Y\geq Z}\right] +\mathbb{E}\left[ Z\varphi _{2}(Y)%
\mathbf{1}_{Y\geq Z}\right] \\
&\leq &\lambda _{1}r_{1}-r_{1}~\mathbb{E}\left[ \varphi _{1}(Z)\right] +%
\mathbb{E}\left[ \left( Z\varphi _{2}(Y)+Z\varphi _{1}(Z)-\Phi (Z)\right) 
\mathbf{1}_{Y\geq Z}\right] \\
&=&\lambda _{1}r_{1}-r_{1}~\mathbb{E}\left[ \varphi _{1}(Z)\right] +\mathbb{E%
}\left[ \left( \Lambda \varphi _{2}(Y)+\Lambda \varphi _{1}(Z)-\Phi (\Lambda
)\right) \mathbf{1}_{Y\geq Z}\right] ,
\end{eqnarray*}%
where we put $\Lambda :=\min \{Y,Z\}.$

Consider next the $Z>Y.$ Interchanging Y and Z and using $Z>y$ instead of $%
Z\geq y$ throughout gives%
\begin{equation*}
\mathbb{E}\left[ s(Y,Z)\mathbf{1}_{Z>Y}\right] \leq \lambda _{2}r_{2}-r_{2}~%
\mathbb{E}\left[ \varphi _{2}(Y)\right] +\mathbb{E}\left[ \left( \Lambda
\varphi _{1}(Z)+\Lambda \varphi _{2}(Y)-\Phi (\Lambda )\right) \mathbf{1}%
_{Z>Y}\right] .
\end{equation*}%
Adding the last two inequalities yields%
\begin{eqnarray}
\mathbb{E}\left[ s(Y,Z)\right] &\leq &\lambda _{1}r_{1}+\lambda
_{2}r_{2}-r_{1}~\mathbb{E}\left[ \varphi _{1}(Z)\right] -r_{2}~\mathbb{E}%
\left[ \varphi _{2}(Y)\right]  \notag \\
&&+\mathbb{E}\left[ \Lambda \varphi _{1}(Z)+\Lambda \varphi _{2}(Y)-\Phi
(\Lambda )\right]  \notag \\
&=&\lambda _{1}r_{1}+\lambda _{2}r_{2}  \notag \\
&&+\mathbb{E}\left[ \varphi _{1}(Z)\left( \Lambda -r_{1}\right) \right] +%
\mathbb{E}\left[ \varphi _{2}(Y)\left( \Lambda -r_{2}\right) \right] -%
\mathbb{E}\left[ \Phi (\Lambda )\right] .  \label{eq:E[s]}
\end{eqnarray}

Now we have%
\begin{eqnarray}
\mathbb{E}\left[ \varphi _{1}(Z)\left( \Lambda -r_{1}\right) \right]
&=&\int_{0}^{\infty }\varphi _{1}(z)(\mathbb{E}\left[ \min \{Y,z\}\right]
-r_{1})f_{2}(z)\mathrm{d}z  \notag \\
&=&\int_{0}^{\infty }\varphi _{1}(z)(H_{1}(z)-r_{1})f_{2}(z)\mathrm{d}z
\label{eq:phi1}
\end{eqnarray}%
(use $\Lambda =\min \{Y,Z\}$ and (\ref{eq:H})).$\mathbb{\ }$Similarly,%
\begin{equation}
\mathbb{E}\left[ \varphi _{2}(Y)\left( \Lambda -r_{2}\right) \right]
=\int_{0}^{\infty }\varphi _{2}(y)(H_{2}(y)-r_{2})f_{1}(y)\mathrm{d}y.
\label{eq:phi2}
\end{equation}%
Let $F_{\Lambda }$ be the cumulative distribution function of $\Lambda =\min
\{Y,Z\};$ then\linebreak\ $1-F_{\Lambda }(u)=G_{\Lambda
}(u)=G_{1}(u)G_{2}(u),$ and 
\begin{eqnarray}
\mathbb{E}\left[ \Phi (\Lambda )\right] &=&\int_{0}^{\infty }\Phi (u)\mathrm{%
d}F_{\Lambda }(u)=-\int_{0}^{\infty }\Phi (u)\mathrm{d}G_{\Lambda }(u) 
\notag \\
&=&\left[ -\Phi (u)G_{\Lambda }(u)\right] _{0}^{\infty }+\int_{0}^{\infty
}\Phi ^{\prime }(u)G_{\Lambda }(u)\mathrm{d}u  \notag \\
&=&\int_{0}^{\infty }\Phi ^{\prime }(u)G_{\Lambda }(u)\mathrm{d}u  \notag \\
&=&\int_{0}^{\infty }(\varphi _{1}(u)+\varphi _{2}(u))G_{1}(u)G_{2}(u)%
\mathrm{d}u,  \label{eq:PHI}
\end{eqnarray}%
where we integrated by parts to get the second line,\footnote{%
Formally, we integrate by parts on a finite interval $[0,M]$ and then let $%
M\rightarrow \infty .$ The functions $G_{\Lambda }$ and $\Phi $ are
absolutely continuous (because $G_{i}=1-F_{i}$ for $i=1,2$ are absolutely
continuous and $G_{\Lambda }=G_{1}G_{2}$, and $\Phi $ is convex and
continuous).} and then used $\Phi (0)=0$ and $\Phi (\infty )G_{\Lambda
}(\infty )=0$ (because $0\leq \Phi (u)G_{\Lambda }(u)\leq
2u(r_{1}/u)(r_{2}/u)\rightarrow 0$ as $u\rightarrow \infty ,$ with $\Phi
(u)\leq 2u$ following from $\Phi ^{\prime }(u)\leq 2).$

Substituting (\ref{eq:phi1})--(\ref{eq:PHI}) into (\ref{eq:E[s]}) yields the
result.
\end{proof}

\section{Bounding the $K_{i}$-Term \label{s:K1-K2}}

In this section we bound from above the term $\int (\varphi
_{1}K_{1}+\varphi _{2}K_{2})$ in (\ref{eq:revenue}) over all possible
functions $\varphi _{i},$ which take values in $[0,\lambda _{i}],$ and whose 
\emph{sum} $\varphi _{1}+\varphi _{2}$ is nondecreasing. This term is linear
in the $\varphi _{i},$ and so, if each $\varphi _{i}$ were nondecreasing, it
would suffice to consider only the extreme functions that take the values $0$
and $\lambda _{i}$ (because any nondecreasing function is an average of such
functions; see the remark below). However, we only require the sum to be
nondecreasing, which requires a more delicate analysis; see Proposition \ref%
{p:I-lambda}. This result is then applied to our specific functions $K_{1}$
and $K_{2}$ to get the bound in Proposition \ref{p:bound-I} (this
constitutes the core of the proof).

From now on we will assume without loss of generality that $\lambda _{1}\leq
\lambda _{2},$ and so $0<\lambda _{1}\leq \lambda _{2}\leq 1.$ Let $%
K_{1},K_{2}:\mathbb{R}_{+}\rightarrow \mathbb{R}$ be two functions, and
define 
\begin{equation*}
I%
{\;:=\;}%
\sup_{\varphi _{1},\varphi _{2}}\int_{0}^{\infty }\left( \varphi
_{1}(t)K_{1}(t)+\varphi _{2}(t)K_{2}(t)\right) \mathrm{d}t,
\end{equation*}%
where the supremum is taken over all functions $\varphi _{i}:\mathbb{R}%
_{+}\rightarrow \lbrack 0,\lambda _{i}]$ such that $\varphi :=\varphi
_{1}+\varphi _{2}$ is a nondecreasing function.

To estimate $I,$ for any $0\leq a\leq b\leq c\leq \infty $ define%
\begin{eqnarray}
I(a,b,c) &%
{\;:=\;}%
&\int_{a}^{b}\lambda _{1}\max \{K_{1}(t),K_{2}(t)\}\mathrm{d}t  \notag \\
&&+\int_{b}^{c}\left( (\lambda _{2}-\lambda _{1})K_{2}(t)+\lambda _{1}\max
\{K_{1}(t),K_{2}(t)\}\right) \mathrm{d}t  \notag \\
&&+\int_{c}^{\infty }(\lambda _{1}K_{1}(t)+\lambda _{2}K_{2}(t))\mathrm{d}t 
\notag \\
&=&\lambda _{1}\int_{a}^{c}\max \{K_{1}(t),K_{2}(t)\}\mathrm{d}t+\lambda
_{1}\int_{c}^{\infty }(K_{1}(t)+K_{2}(t))\mathrm{d}t  \notag \\
&&+(\lambda _{2}-\lambda _{1})\int_{b}^{\infty }K_{2}(t)\mathrm{d}t.
\label{eq:Iabc}
\end{eqnarray}%
It is immediate to see that $I(a,b,c)$ is nothing other than $\int (\varphi
_{1}K_{1}+\varphi _{2}K_{2})$ for the following functions $\varphi _{1}$ and 
$\varphi _{2}$: 
\begin{equation*}
\renewcommand{\arraystretch}{1.5}%
\begin{tabular}{c||c|c|c|}
\cline{2-4}
& $a\leq t<b$ & $b\leq t<c$ & $t\geq c$ \\ \hline\hline
\multicolumn{1}{|c||}{$\varphi _{1}(t)$} & $\lambda _{1}\mathbf{1}%
_{K_{1}(t)\geq K_{2}(t)}$ & $\lambda _{1}\mathbf{1}_{K_{1}(t)\geq K_{2}(t)}$
& $\lambda _{1}$ \\ \hline
\multicolumn{1}{|c||}{$\varphi _{2}(t)$} & $\lambda _{1}\mathbf{1}%
_{K_{1}(t)<K_{2}(t)}$ & $\lambda _{1}\mathbf{1}_{K_{1}(t)<K_{2}(t)}+\lambda
_{2}-\lambda _{1}$ & $\lambda _{2}$ \\ \hline
\end{tabular}%
\renewcommand{\arraystretch}{1}%
\end{equation*}%
Their sum $\varphi _{1}+\varphi _{2}$ then equals%
\begin{equation*}
\renewcommand{\arraystretch}{1.5}%
\begin{tabular}{c||c|c|c|}
\cline{2-4}
& $a\leq t<b$ & $b\leq t<c$ & $t\geq c$ \\ \hline\hline
\multicolumn{1}{|c||}{$\varphi _{1}(t)+\varphi _{2}(t)$} & $\lambda _{1}$ & $%
\lambda _{2}$ & $\lambda _{1}+\lambda _{2}$ \\ \hline
\end{tabular}%
,%
\renewcommand{\arraystretch}{1}%
\end{equation*}%
which is a nondecreasing function, and so $I(a,b,c)\leq I.$

\bigskip

\begin{proposition}
\label{p:I-lambda}Let $0<\lambda _{1}\leq \lambda _{2}\leq 1.$ Then%
\begin{equation*}
I=\sup_{0\leq a\leq b\leq c\leq \infty }I(a,b,c).
\end{equation*}
\end{proposition}

\bigskip

\noindent \textbf{Remark. }We will use the following well-known result.
Every nondecreasing function $\psi :[u,v]\rightarrow \lbrack 0,1]$ (where $%
-\infty \leq u\leq v\leq \infty )$ can be expressed as an (integral) average
of nondecreasing functions that take only the values $0$ and\footnote{%
Assume first that $\psi (v)=1.$ If $\psi $ is a right-continuous function
then $\psi $ may be viewed as a cumulative distribution function on $[u,v]$,
and we have $\psi (t)=\int_{[u,t]}^{{}}\mathrm{d}\psi (x)=\int_{[u,v]}%
\mathbf{1}_{[x,v]}(t)\mathrm{d}\psi (x)$ for every $t\in \lbrack u,v]$
(where $\mathbf{1}_{E}$ is the indicator function of the set $E,$ i.e., $%
\mathbf{1}_{E}(t)=1$ if $t\in E$ and $\mathbf{1}_{E}(t)=0$ otherwise). If $%
\psi $ is not necessarily right-continuous, let $\psi
_{+}(t):=\lim_{t^{\prime }\searrow t}\psi (t^{\prime })$ (which is
right-contiuous), $\psi _{-}(t):=\lim_{t^{\prime }\nearrow t}\psi (t^{\prime
}),$ and take $\lambda _{t}\in \lbrack 0,1]$ such that $\psi (t)=\lambda
_{t}\psi _{+}(t)+(1-\lambda _{t})\psi _{-}(t);$ then $\psi
=\int_{[u,v]}(\lambda _{x}\mathbf{1}_{[x,v]}+\allowbreak (1-\lambda _{x})%
\mathbf{1}_{(x,v]})d\psi _{+}(x).$%
\par
If $0<\psi (v)<1$ then $\psi =\psi (v)\tilde{\psi}+(1-\psi (v))\mathbf{0,}$
where $\tilde{\psi}(t):=\psi (t)/\psi (v)$ and $\mathbf{0}$ is the zero
function (i.e., $\mathbf{0}(t)=0$ for all $t),$ and we apply the above to $%
\tilde{\psi}.$ Finally, if $\psi (v)=0$ then $\psi =\mathbf{0.}$} $1.$ More
generally, every nondecreasing function $\psi :[u,v]\rightarrow \lbrack
\alpha ,\beta ]$ (where $\alpha \leq \beta $ are finite) can be expressed as
an average of nondecreasing functions that take only the two values $\alpha $
and $\beta $ (when $\alpha <\beta ,$ apply the above to $(\psi -\alpha
)/(\beta -\alpha ),$ which takes values in $[0,1]).$ Therefore, when we
maximize a linear functional $\int_{u}^{v}\psi (t)K(t)\mathrm{d}t$ over all
nondecreasing functions $\psi :[u,v]\rightarrow \lbrack \alpha ,\beta ],$ it
suffices to consider those functions that take only the two extreme values $%
\alpha $ and $\beta .$

\bigskip

\begin{proof}
We have seen above that $I\geq I(a,b,c)$ for every $a,b,c.$ We now show that
the supremum in $I$ cannot be higher.

For each $t,$ given $\varphi (t)=\varphi _{1}(t)+\varphi _{2}(t),$ the
expression $\varphi _{1}(t)K_{1}(t)+\varphi _{2}(t)K_{2}(t)$ is maximized by
putting as much weight as possible---subject to the constraints $0\leq
\varphi _{i}(t)\leq \lambda _{i}$---on the higher of $K_{1}(t)$ and $%
K_{2}(t).$ This gives the following upper bounds on $\varphi
_{1}(t)K_{1}(t)+\varphi _{2}(t)K_{2}(t)$:

\begin{itemize}
\item $\varphi (t)\max \{K_{1}(t),K_{2}(t)\}$ for every $t$ in the interval
where $0\leq \varphi (t)\leq \lambda _{1};$

\item $(\varphi (t)-\lambda _{1})K_{2}(t)+\lambda _{1}\max
\{K_{1}(t),K_{2}(t)\}$ for every $t$ in the interval where $\lambda _{1}\leq
\varphi (t)\leq \lambda _{2}$ (because $\varphi _{1}(t)\leq \lambda _{1}$
implies $\varphi _{2}(t)\geq \varphi (t)-\lambda _{1});$ and

\item $(\varphi (t)-\lambda _{2})K_{1}(t)+(\varphi (t)-\lambda
_{1})K_{2}(t)+ $ $(\lambda _{1}+\lambda _{2}-\varphi (t))\max
\{K_{1}(t),K_{2}(t)\}$ for every $t$ in the interval where $\lambda _{2}\leq
\varphi (t)\leq \lambda _{1}+\lambda _{2}.$
\end{itemize}

In each one of these three intervals the bound is affine in $\varphi $ and
so, by the remark above, when maximizing over nondecreasing $\varphi ,$ it
suffices to consider solely those functions $\varphi $ that take only the
corresponding two extreme values. Altogether, such a $\varphi $ takes only
the values $0,\lambda _{1},\lambda _{2},$ and $\lambda _{1}+\lambda _{2}$,
say on the intervals $(0,a),(a,b),(b,c),$ and $(c,\infty ),$
respectively---and then $\int (\varphi _{1}K_{1}+\varphi _{2}K_{2})$ becomes
precisely $I(a,b,c)$. Thus indeed $I\leq \sup I(a,b,c).$
\end{proof}

\bigskip

We now come to the main argument of our proof, which yields, using
Proposition \ref{p:I-lambda}, an upper bound on the $K_{i}$-term for our
specific functions $K_{i}$.

\begin{proposition}
\label{p:bound-I}Let $0<\lambda _{1}\leq \lambda _{2}\leq 1,$ and let $%
K_{1},K_{2}$ be given by (\ref{eq:K1}) and (\ref{eq:K2}). Then%
\begin{equation*}
I\leq \frac{1}{e}\left( \lambda _{2}r_{1}+\lambda _{1}r_{2}+\lambda
_{1}(e-1)\min \{r_{1},r_{2}\}\right) .
\end{equation*}
\end{proposition}

\begin{proof}
Recalling (\ref{eq:H}), we have the following: for each\textbf{\ }$i=1,2,$
the function $H_{i}(t)$ is continuous and strictly increasing at each $t$\
in the support of $X_{i}$ (because $G_{i}(t)>0$ there), and $H_{i}(\infty )=%
\mathbb{E}\left[ X_{i}\right] \geq r_{i}$ (with strict inequality unless $%
X_{i}$ is constant, in which case everything trivializes). Therefore there
exists a finite $\tau _{i}$ such that $H_{i}(\tau _{i})=r_{i};$ since for
all $t<r_{i}$ we have $H_{i}(t)\leq t<r_{i}$ (because $G_{i}\leq 1),$ it
follows that 
\begin{equation*}
\tau _{i}\geq r_{i}.
\end{equation*}

Put $L_{i}(t):=G_{j}(t)(H_{i}(t)-r_{i});$ taking derivatives gives 
\begin{equation*}
L_{i}^{\prime }(t)=-f_{j}(t)(H_{i}(t)-r_{i})+G_{j}(t)G_{i}(t)=-K_{i}(t).
\end{equation*}

We will use the following estimates:%
\begin{equation}
\int_{u}^{\infty }G_{1}(t)G_{2}(t)\mathrm{d}t\leq \int_{u}^{\infty }\frac{%
r_{1}}{t}\frac{r_{2}}{t}\mathrm{d}t=\frac{r_{1}r_{2}}{u}  \label{eq:int-G1G2}
\end{equation}%
for every $u>0$ (because $G_{i}(t)\leq r_{i}/t);$%
\begin{equation*}
H_{i}(u)-r_{i}\leq r_{i}\log \frac{u}{r_{i}}
\end{equation*}%
for every $u\geq r_{i}$ (recall (\ref{eq:bound-H-0})); and, thus, 
\begin{equation}
L_{i}(u)=G_{j}(u)(H_{i}(u)-r_{i})\leq \frac{r_{i}r_{j}}{u}\log \frac{u}{r_{i}%
}  \label{eq:bound-L}
\end{equation}%
for every $u\geq r_{i}.$ The last inequality implies that $%
L_{i}(u)\rightarrow 0$ as $u\rightarrow \infty ,$ and so 
\begin{equation}
\int_{u}^{\infty }K_{i}(t)\mathrm{d}t=\left[ -L_{i}(t)\right] _{u}^{\infty
}=L_{i}(u).  \label{eq:int-K}
\end{equation}%
Finally, letting $\{i,j\}=\{1,2\},$\ we have%
\begin{eqnarray}
L_{i}(u) &\leq &\frac{1}{e}r_{j}\text{\ \ \ and}  \label{eq:L<1/e} \\
L_{i}(u)+\frac{r_{i}r_{j}}{u} &\leq &r_{j}  \label{eq:L+rr/t}
\end{eqnarray}%
for every $u\geq r_{i}$ (use (\ref{eq:bound-L}) together with $\log x/x\leq
1/e$ and $(\log x+1)/x\leq 1$ for all $x>0;$ note that there is no typo
here: these bounds on $L_{i}$ use $r_{j}$ rather than $r_{i}).$

We need to bound $I(a,b,c).$ For the last term of (\ref{eq:Iabc}) we have,
by (\ref{eq:int-K}) and (\ref{eq:L<1/e}), 
\begin{equation}
\int_{b}^{\infty }K_{2}(t)\mathrm{d}t\leq \frac{1}{e}r_{1},
\label{eq:c-term}
\end{equation}%
and so it remains to estimate $J(a,c):=\int_{a}^{c}\max
\{K_{1},K_{2}\}+\int_{c}^{\infty }(K_{1}+K_{2}).$ A main difficulty in doing
so is that the $K_{i}$ are neither nonnegative nor monotonic, and may change
signs many times. To handle this we define for each $i$ an auxiliary
function $M_{i}(t):=f_{j}(t)(H_{i}(t)-r_{i})=K_{i}(t)+G_{1}(t)G_{2}(t),$
which vanishes at $t=\tau _{i}$, is nonpositive before $\tau _{i},$\ and
nonnegative after $\tau _{i}$; i.e., $M_{i}(t)\geq 0$\ for $t\geq \tau _{i}$%
\ and $M_{i}(t)\leq 0$\ for $t\leq \tau _{i}.$\ 

We distinguish three cases according to the location of $a$\ relative to $%
\tau _{1}$ and $\tau _{2}$ (the points where $M_{1}$ and $M_{2}$ change
sign); without loss of generality\footnote{%
The expression $J(a,b)$ that we estimate now is symmetric in $i=1,2,$ and so
the assumption that $\lambda _{1}\leq \lambda _{2}$ is irrelevant here; we
thus assume that $\tau _{1}\leq \tau _{2}.$} assume that $\tau _{1}\leq \tau
_{2}.$

\noindent $\bullet $ \textbf{Case 1}\emph{. }$a\geq \max \{\tau _{1},\tau
_{2}\}=\tau _{2}.$ For every $t\geq a$ we have $M_{i}(t)\geq 0$ (because $%
t\geq a\geq \tau _{i}),$ and thus\footnote{%
This is the inequality $\max \{x,y\}\leq x+y+z$\ whenever\textbf{\ }$x,y\geq
-z.$}%
\begin{eqnarray*}
\max \{K_{1}(t),K_{2}(t)\} &=&\max \{M_{1}(t),M_{2}(t)\}-G_{1}(t)G_{2}(t) \\
&\leq &M_{1}(t)+M_{2}(t)-G_{1}(t)G_{2}(t) \\
&=&K_{1}(t)+K_{2}(t)+G_{1}(t)G_{2}(t).
\end{eqnarray*}%
Since we clearly also have $K_{1}+K_{2}\leq K_{1}+K_{2}+G_{1}G_{2},$ we get%
\begin{eqnarray*}
J(a,c) &\leq &\int_{a}^{\infty }(K_{1}(t)+K_{2}(t)+G_{1}(t)G_{2}(t))\mathrm{d%
}t \\
&\leq &L_{1}(a)+L_{2}(a)+\frac{r_{1}r_{2}}{a}=:\bar{J}(a)
\end{eqnarray*}%
by (\ref{eq:int-K}) and (\ref{eq:int-G1G2}). If, say, $r_{k}\leq r_{\ell }$
(where $\{k,\ell \}=\{1,2\})$ then using (\ref{eq:L<1/e}) for $i=k$ and (\ref%
{eq:L+rr/t}) for $i=\ell $ (recall that $a\geq \tau _{i}\geq r_{i}$\ for
both $i)$ yields\footnote{%
A slightly better estimate of $(2/\sqrt{e})\sqrt{r_{1}r_{2}}$ may be
obtained here by directly maximizing $\bar{J}(a)$ over $a;$ however, this
will not improve the overall estimate, due to Cases 2 and 3.}%
\begin{equation*}
J(a,c)\leq \bar{J}(a)\leq \frac{1}{e}r_{\ell }+r_{k}.
\end{equation*}

\noindent $\bullet $ \textbf{Case 2}\emph{. }$\tau _{1}\leq a<\tau _{2}.$ In
the range $t\in \lbrack a,\tau _{2})\subseteq \lbrack \tau _{1},\tau _{2})$
we have $M_{1}(t)\geq 0\geq M_{2}(t),$ and so $K_{1}(t)\geq K_{2}(t)$ and $%
K_{2}(t)\leq 0;$ therefore both $\max \{K_{1}(t),K_{2}(t)\}$ and $%
K_{1}(t)+K_{2}(t)$ are $\leq K_{1}(t),$ and thus, regardless of where $c$ is,%
\begin{eqnarray*}
J(a,c) &\leq &\int_{a}^{\tau _{2}}K_{1}(t)\mathrm{d}t+\bar{J}(\tau
_{2})=\left( L_{1}(a)-L_{1}(\tau _{2})\right) +\left( L_{1}(\tau
_{2})+L_{2}(\tau _{2})+\frac{r_{1}r_{2}}{\tau _{2}}\right) \\
&=&L_{1}(a)+\frac{r_{1}r_{2}}{\tau _{2}}\leq \frac{1}{e}r_{2}+\min
\{r_{1},r_{2}\}\leq \frac{1}{e}r_{\ell }+r_{k},
\end{eqnarray*}%
where we have used: $L_{2}(\tau _{2})=0$ and $\tau _{2}=\max \{\tau
_{1},\tau _{2}\}\geq \max \{r_{1},r_{2}\}$ (because $\tau _{i}\geq r_{i}$).

\noindent $\bullet $ \textbf{Case 3}\emph{. }$a<\min \{\tau _{1},\tau
_{2}\}=\tau _{1}.$ For every $t\leq \tau _{1}$ we have $K_{1}(t)\leq
M_{1}(t)\leq 0$ and $K_{2}(t)\leq M_{2}(t)\leq 0,$ and so both $\max
\{K_{1}(t),K_{2}(t)\}$ and $K_{1}(t)+K_{2}(t)$ are $\leq 0$ in the interval $%
[a,\tau _{1}].$ Therefore $J(a,c)\leq J(\tau _{1},\max \{c,\tau _{1}\}),$ to
which we apply Case 2 (with $a=\tau _{1})$.

Thus in all three cases the bound on $J(a,c)$ is $(1/e)r_{\ell
}+r_{k}=(1/e)(r_{1}+r_{2})+(1-1/e)\min \{r_{1},r_{2}\};$ together with (\ref%
{eq:c-term}), we get%
\begin{equation*}
I(a,b,c)\leq \frac{\lambda _{1}}{e}r_{1}+\frac{\lambda _{1}}{e}r_{2}+\frac{%
\lambda _{1}(e-1)}{e}\min \{r_{1},r_{2}\}+\frac{\lambda _{2}-\lambda _{1}}{e}%
r_{1},
\end{equation*}%
completing the proof.
\end{proof}

\bigskip

\noindent \textbf{Remark. }If $\varphi _{1}$ and $\varphi _{2}$ are \emph{%
each} required to be nondecreasing (rather than just their sum), then we get
a smaller bound on $\int \left( \varphi _{1}K_{1}+\varphi _{2}K_{2}\right) $%
, namely:%
\begin{eqnarray*}
&&\sup_{\varphi _{1},\varphi _{2}}\int_{0}^{\infty }\left( \varphi
_{1}(t)K_{1}(t)+\varphi _{2}(t)K_{2}(t)\right) \mathrm{d}t \\
&=&\sup_{0\leq a\leq \infty }\lambda _{1}\int_{a}^{\infty
}K_{1}(t)dt+\sup_{0\leq b\leq \infty }\lambda _{2}\int_{b}^{\infty
}K_{2}(t)dt\leq \frac{\lambda _{1}}{e}r_{2}+\frac{\lambda _{2}}{e}r_{1}
\end{eqnarray*}%
(use the remark preceding the proof of Proposition \ref{p:I-lambda} together
with (\ref{eq:int-K}) and (\ref{eq:L<1/e})). Therefore, for mechanisms $\mu
=(q,s)$ where $q_{1}(t,t)$ and $q_{2}(t,t)$ are monotonic---such as, for
instance, symmetric mechanisms, where\footnote{%
This proves Theorem B of Hart and Nisan (2017) for two independent \emph{and
identically distributed }goods.} $q_{1}(t,t)=q_{2}(t,t)$---we get, taking $%
\lambda _{1}=\lambda _{2}=1$ in Proposition \ref{p:revenue},%
\begin{equation*}
R(\mu ;X)\leq r_{1}+r_{2}+\frac{1}{e}r_{1}+\frac{1}{e}r_{2}=\left( 1+\frac{1%
}{e}\right) (r_{1}+r_{2}).
\end{equation*}%
This yields the bound $e/(e+1),$ which is better than $\sqrt{e}/(\sqrt{e}%
+1). $

\section{Completing the Proof\label{s:proof of thm A}}

Combining the results of the previous two sections yields the first part of
our Main Result:

\begin{theorem}
\label{th:62.2}Let $X=(X_{1},X_{2})$ be a two-good random valuation with
independent goods. Then%
\begin{equation*}
\frac{\text{\textsc{SRev}}(X_{1},X_{2})}{\text{\textsc{Rev}}(X_{1},X_{2})}%
\geq \frac{\sqrt{e}}{\sqrt{e}+1}\approx 0.62.
\end{equation*}
\end{theorem}

\begin{proof}
Let\footnote{%
The results of the previous sections will be applied to the rescaled $\tilde{%
X}_{i}=X_{i}/\lambda _{i},$ and so we will use $r_{i}$ for the revenue of $%
\tilde{X}_{i},$ and $R_{i}$ for the revenue of the original $X_{i}.$} $%
R_{i}:=$\textsc{Rev}$(X_{i});$ thus \textsc{SRev}$(X_{1},X_{2})=R_{1}+R_{2}.$
Given $0<\lambda _{1}\leq \lambda _{2},$ put $\tilde{X}_{i}:=X_{i}/\lambda
_{i}$ and $r_{i}:=$\textsc{Rev}$(\tilde{X}_{i})=R_{i}/\lambda _{i}.$ Using
Lemma \ref{l:rescale-rev}, Proposition \ref{p:revenue} for $(\tilde{X}_{1},%
\tilde{X}_{2}),$ and then Proposition \ref{p:bound-I}, yields%
\begin{eqnarray}
\text{\textsc{Rev}}(X_{1},X_{2}) &\leq &\lambda _{1}\frac{R_{1}}{\lambda _{1}%
}+\lambda _{2}\frac{R_{2}}{\lambda _{2}}+\frac{\lambda _{2}}{e}\frac{R_{1}}{%
\lambda _{1}}+\frac{\lambda _{1}}{e}\frac{R_{2}}{\lambda _{2}}+\frac{\lambda
_{1}(e-1)}{e}\min \left\{ \frac{R_{1}}{\lambda _{1}},\frac{R_{2}}{\lambda
_{2}}\right\}  \notag \\
&\leq &R_{1}+R_{2}+\frac{1}{\lambda e}R_{1}+\frac{\lambda }{e}R_{2}+\frac{%
\lambda (e-1)}{e}R_{2}  \label{eq:R1R2} \\
&=&R_{1}+R_{2}+\frac{1}{\lambda e}R_{1}+\lambda R_{2},  \notag
\end{eqnarray}%
where in the second line we put $\lambda :=\lambda _{1}/\lambda _{2}\in
(0,1] $ and used $\min \left\{ R_{1}/\lambda _{1},R_{2}/\lambda _{2}\right\}
\leq R_{2}/\lambda _{2}.$ The final expression equals $(1+1/\sqrt{e}%
)(R_{1}+R_{2}) $ when $\lambda =1/\sqrt{e},$ completing the proof.\footnote{%
One may check that $1+1/\sqrt{e}$ is the best bound that is independent of $%
R_{1}$ and $R_{2}$ (when $R_{1}=R_{2}$ the above expression is minimized $%
only$ at $\lambda =1/\sqrt{e}).$}
\end{proof}

\section{Regular Goods and Monotonic Mechanisms\label{s:73}}

In this section we prove the second part of our Main Result, namely, the
better bound of $73\%$ for regular goods (and also for monotonic
mechanisms). We will use here only the symmetric diagonal decomposition
(i.e., $\lambda _{1}=\lambda _{2}=1).$

Following Myerson (1981), we say that a one-dimensional random variable $X$
is $\emph{weakly}$ \emph{regular} if its support is an interval $[\alpha
,\beta ]$ with\footnote{%
Notice that we allow $\beta =\infty ,$ in which case the interval is
understood to be $[\alpha ,\infty ).$} $0\leq \alpha <\beta \leq \infty ,$
on which it has a density function $f(t)$ that is positive and continuous,
and the resulting \textquotedblleft virtual valuation function" $%
t-G(t)/f(t)\ $is nondecreasing (Myerson's regularity condition requires the
virtual valuation to be \emph{strictly} increasing).

\begin{lemma}
\label{l:regular}Assume that $X_{1}$ and $X_{2}$ are weakly regular. Then $%
K_{i}(u)>0$ implies that $K_{i}(v)\geq 0$ for all $v>u,$ for $i=1,2.$
\end{lemma}

\begin{proof}
Let $[\alpha _{i},\beta _{i}]$ be the support of $X_{i}.$ Assume by way of
contradiction that, say, $K_{1}(u)>0$ and $K_{1}(v)<0$ for some $v>u.$
First, $K_{1}(u)>0$ implies that $f_{2}(u)>0$ and $H_{1}(u)-r_{1}>0$
(otherwise $K_{1}(u)\leq -G_{1}(u)G_{2}(u)\leq 0),$ and so $\alpha _{2}\leq
u\leq \beta _{2}$ and $u>\alpha _{1}$ (because $H_{1}$ is nondecreasing and $%
H_{1}(\alpha _{1})\leq \alpha _{1}=\alpha _{1}\cdot G_{1}(\alpha _{1})\leq
r_{1}).$ Second, $K_{1}(v)<0$ implies that $G_{1}(v)>0$ and $G_{2}(v)>0$
(otherwise $K_{1}(v)=f_{2}(v)(H_{1}(v)-r_{1})\geq
f_{2}(v)(H_{1}(u)-r_{1})\geq 0$ since $H_{1}$ is nondecreasing), and so $%
v<\beta _{1}$ and $v<\beta _{2}$ (because $G_{i}(\beta _{i})=0).$ Together
with $u<v$ it follows that $u$ and $v$ both lie in the interval\footnote{%
I.e., $t_{1}<u<v<\beta _{1},$ where $t_{1}>\alpha _{1}$ is the point where $%
H_{1}(t_{1})=r_{1},$ and $\alpha _{2}\leq u<v<\beta _{2}.$} where $%
f_{2}(t)>0,$ $G_{1}(t)>0,$ and $H_{1}(t)-r_{1}>0.$ But in that interval the
function $\kappa ,$ defined by%
\begin{equation*}
\kappa (t):=\frac{K_{1}(t)}{f_{2}(t)G_{1}(t)}=\left( \frac{H_{1}(t)-r_{1}}{%
G_{1}(t)}-t\right) +\left( t-\frac{G_{2}(t)}{f_{2}(t)}\right) ,
\end{equation*}%
is increasing---the derivative of the first term is $%
f_{1}(t)(H_{1}(t)-r_{1})/G_{1}^{2}(t)>0,$ and the second term is
nondecreasing by regularity. Therefore we cannot have $\kappa (u)>0$\ and $%
\kappa (v)<0,$\ which contradicts the assumption that $K_{1}(u)>0$\ and $%
K_{1}(v)<0.$
\end{proof}

\bigskip

This yields the second part of our Main Result:

\begin{theorem}
\label{th:regular}Let $X=(X_{1},X_{2})$ be a two-good random valuation with
independent and weakly regular\emph{\ }goods. Then%
\begin{equation*}
\frac{\text{\textsc{SRev}}(X_{1},X_{2})}{\text{\textsc{Rev}}(X_{1},X_{2})}%
\geq \frac{e}{e+1}\approx 0.73.
\end{equation*}
\end{theorem}

\begin{proof}
Take $\lambda _{1}=\lambda _{2}=1$ and let $r_{i}=$\textsc{Rev}$(X_{i}).$
Lemma \ref{l:regular} implies that if $K_{i}(t)$ is positive anywhere then
it is nonnegative from that point on, and so either (i) there is some finite 
$u\geq 0$ such that $K_{i}(t)\leq 0$ for $t<u$ and $K_{i}(t)\geq 0$ for $%
t\geq u$ or (ii) $K_{i}(t)\leq 0$ for all $t\geq 0.$ Therefore, for any
function $\varphi _{i}$ with values in $[0,1],$ we have%
\begin{equation*}
\int_{0}^{\infty }\varphi _{i}(t)K_{i}(t)\mathrm{d}t\leq \int_{u}^{\infty
}K_{i}(t)\mathrm{d}t=L_{i}(u)\leq \frac{1}{e}r_{j}
\end{equation*}%
in case (i) (by (\ref{eq:int-K}) and (\ref{eq:L<1/e})), and $%
\int_{0}^{\infty }\varphi _{i}(t)K_{i}(t)$\textrm{$d$}$t\leq 0$ in case
(ii). Altogether $\int \varphi _{1}K_{1}+\int \varphi _{2}K_{2}\leq
(1/e)r_{2}+(1/e)r_{1},$\textbf{\ }and so Proposition \ref{p:revenue} gives 
\textsc{Rev}$(X)\leq (1+1/e)(r_{1}+r_{2}),$\textbf{\ }proving the result.
\end{proof}

\bigskip

Next, let \textsc{MonRev}$(X)$ denote the maximal revenue that can be
obtained using \emph{monotonic }mechanisms, i.e., mechanisms $\mu =(q,s)$
for which the function $s(x)\ $is nondecreasing in $x.$

\begin{proposition}
Let $X=(X_{1},X_{2})$ be a two-good random valuation with independent and 
\emph{weakly regular }goods. Then%
\begin{equation*}
\frac{\text{\textsc{SRev}}(X_{1},X_{2})}{\text{\textsc{MonRev}}(X_{1},X_{2})}%
\geq \frac{e}{e+1}\approx 0.73.
\end{equation*}
\end{proposition}

\begin{proof}
Put $r_{i}:=$\textsc{Rev}$(X_{i}),$ and let $V_{i}$ be the \textquotedblleft
equal-revenue" (\textrm{$\mathtt{ER}$}) random valuation with the same
revenue $r_{i}$ as $X_{i};$ i.e., its tail distribution function is $\hat{G}%
_{i}(t)=\min \{r_{i}/t,1\}\geq G_{i}(t)$. Take $V_{1}$ and $V_{2}$ to be
independent, and put $V=(V_{1},V_{2}).$ Because $V_{i}$ first-order
stochastically dominates $X_{i},$ for every monotonic mechanism $\mu =(q,s)$
we have $R(\mu ;X)=\mathbb{E}\left[ s(X_{1},X_{2})\right] \leq \mathbb{E}%
\left[ s(V_{1},V_{2})\right] =R(\mu ;V).$ Therefore 
\begin{equation*}
\text{\textsc{MonRev}}(X)\leq \text{\textsc{MonRev}}(V)\leq \text{\textsc{Rev%
}}(V).
\end{equation*}%
The \textrm{$\mathtt{ER}$}-good\textrm{\ }$V_{i}$ is weakly regular (because
on its support $[r_{i},\infty )$ the virtual valuation function $t-\hat{G}%
_{i}(t)/\hat{f}_{i}(t)$ is identically $0),$ and so \textsc{SRev}$(V)\geq
e/(e+1)$\textsc{Rev}$(V)$ by Theorem \ref{th:regular}; together with \textsc{%
SRev}$(X)=\ $\textsc{SRev}$(V)$ (by construction) and \textsc{MonRev}$%
(X)\leq \ $\textsc{Rev}$(V)$ (see above), the result follows.
\end{proof}

\appendix{}

\section{Appendix: Revenue Continuity\label{s:continuity}}

This appendix deals with the continuity of the revenue with respect to
valuations, which is of independent interest. Take a sequence of $k$-good
valuations $X^{n}$ that converges in distribution to the $k$-good random
valuation $X$; does the sequence of revenues \textsc{Rev}$(X^{n})$ converge
to\footnote{%
Only the distribution of a random valuation $X$ matters for the revenue
achievable from $X;$ it is thus natural to consider what happens when $X^{n}$
converges in distribution to $X$. Formally, convergence in distribution is
equivalent to the cumulative distribution functions converging pointwise at
all points of continuity of the limit cumulative distribution. Informally,
being close in distribution means that the probabilities of nearby values
are close (see, for instance, (\ref{eq:prohorov}) below). Billingsley's
(1968) book is a good reference for the concepts used here.} \textsc{Rev}$%
(X) $?

Even in the one-good case this need not be so: for each $n$ let $X^{n}$ be
the one-good valuation that takes value $0$ with probability $1-1/n$ and
value $n$ with probability $1/n.$ Then $X^{n}$ converges in distribution to
the valuation $X$ that takes value $0$ with probabilty $1.$ But \textsc{Rev}$%
(X^{n})=1$ (with the posted price of $n)$ while \textsc{Rev}$(X)=0.$

We will show that if the valuations all lie in a bounded set---more
generally, if the random valuations are uniformly integrable---then the
limit of the revenues equals the revenue of the limit. We emphasize that all
the results in this appendix are for \emph{general} $k$-good valuations for
any\textbf{\ }$k\geq 1$, whether the goods are independent or not.\footnote{%
Monteiro (2015) establishes continuity of the optimal revenue in the
one-good case with $n$ independent buyers, when the valuations are bounded
and the limit distributions are continuous (his proof uses the
characterization of the optimal mechanism).}

Some notation. First, it will be convenient to work here with the $\ell _{1}$%
-norm on\footnote{%
This affects only the various constants below.} $\mathbb{R}^{k},$ i.e., $%
||x||_{1}=\sum_{i=1}^{k}|x_{i}|.$ The $\ell _{1}$-norm of a valuation $x$ in 
$\mathbb{R}_{+}^{k}$ provides a simple bound on the seller's payoff in any
mechanism $\mu =(q,s)$ that is IR: $s(x)\leq q(x)\cdot x\leq
\sum_{i}x_{i}=||x||_{1}$; thus, if a random valuation $X$ satisfies $%
||X||_{1}\leq M$ then \textsc{Rev}$(X)\leq M.$ Second, the \emph{Prohorov
distance} between the distributions of $X$ and $Y,$ which we denote by $%
\mathrm{Dist}(X,Y),$ is defined as the infimum of all $\rho >0$ such that%
\begin{eqnarray}
\mathbb{P}\left[ X\in A\right] &\leq &\mathbb{P}\left[ Y\in B_{\rho }(A)%
\right] +\rho ,\text{\ \ and}  \label{eq:prohorov} \\
\mathbb{P}\left[ Y\in A\right] &\leq &\mathbb{P}\left[ X\in B_{\rho }(A)%
\right] +\rho  \notag
\end{eqnarray}%
for all measurable sets $A,$ where $B_{\rho }(A):=\{y:||y-x||_{1}<\rho $ for
some $x\in A\}$ is the $\rho $-neighborhood of $A.$ Thus $0\leq \mathrm{Dist}%
(X,Y)\leq 1,$ and $X^{n}$ converges in distribution to $X,$ which we write
as $X^{n}\overset{\mathcal{D}}{\rightarrow }X,$ if and only if $\mathrm{Dist}%
(X,X^{n})\rightarrow 0$ (again, see Billingsley 1968).

The basic result is that in the bounded case the distance between the
revenues of two random valuations is uniformly bounded by a function of the
Prohorov distance between their distributions.

\begin{proposition}
\label{p:dist}Let $X$ and $Y$ be $k$-good valuations with bounded values,
say, $||X||_{1},||Y||_{1}\leq M$ for some $M\geq 1.$ Then\footnote{%
We have not attempted to optimize the bound here.}%
\begin{equation*}
|\text{\textsc{Rev}}(X)-\text{\textsc{Rev}}(Y)|\leq (2M+1)\sqrt{\mathrm{Dist}%
(X,Y)}.
\end{equation*}
\end{proposition}

\begin{proof}
If $\mathrm{Dist}(X,Y)=1$ there is nothing to prove, since both revenues are
between $0$ and $M.$

Let thus $0<\rho <1$ be such that (\ref{eq:prohorov}) holds for every
measurable set $A\subseteq \mathbb{R}_{+}^{k},$ and take $\alpha $ so that $%
\rho \leq \alpha <1$ (the value of $\alpha $ will be determined later).
Denote by $D_{M}:=\{x\in \mathbb{R}_{+}^{k}:||x||_{1}\leq M\}$ the domain of
values of $X$ and $Y.$

Let $\mu =(q,s)$ be an IC, IR, and NPT mechanism, and let $b$ be its buyer
payoff function. We define a new mechanism $\tilde{\mu}$ by lowering all
payments by a factor of $1-\alpha $ (and letting the buyer reoptimize).
Thus, let $\mathrm{cl~}W\subset \mathbb{R}_{+}^{k+1}$ be the closure of the
set $W:=\{(q(x),(1-\alpha )s(x)):x\in D_{M}\}.$ For each $x\in D_{M}$ let $(%
\tilde{q}(x),\tilde{s}(x))$ be a maximizer for\footnote{%
This maximum is attained because $W$ is bounded, namely, $W\subseteq \lbrack
0,1]^{k}\times \lbrack 0,M],$ and so $\mathrm{cl}\ W$ is a compact set.} $%
\tilde{b}(x):=\max_{(g,t)\in \mathrm{cl}\ W}(g\cdot x-t).$ Then the
mechanism $\tilde{\mu}=(\tilde{q},\tilde{s})$ is IC (by the maximizer
definition), IR (because $\tilde{b}(x)\geq b(x)+\alpha s(x),$ which is
nonegative since $\mu $ is IR and NPT), and NPT (because $\mu $ is NPT).%
\footnote{%
Hart and Reny (2015) use this device of applying a small uniform discount to
the buyer's payments to show that, at an arbitrarily small cost, the seller
can perturb any IC and IR mechanism so that the buyer breaks any
indifference in the seller's favor (the resulting mechanism is thus called 
\emph{seller-favorable}).}

Let $x,y\in D_{M}$ be such that $||x-y||_{1}\leq \rho .$ Then $(\tilde{q}(y),%
\tilde{s}(y))\in \mathrm{cl}\ W$ can be approximated by elements of $W$: for
every $\varepsilon >0$ there is $z\in D_{M}$ such that, in particular, $|%
\tilde{s}(y)-(1-\alpha )s(z)|\leq \varepsilon $ and $|[\tilde{q}(y)\cdot y-%
\tilde{s}(y)]-[q(z)\cdot y-(1-\alpha )s(z)]|\leq \varepsilon .$ We then have%
\begin{eqnarray*}
q(z)\cdot y-(1-\varepsilon )s(z)+\varepsilon &\geq &\tilde{q}(y)\cdot y-%
\tilde{s}(y) \\
&\geq &q(x)\cdot y-(1-\alpha )s(x) \\
&=&q(x)\cdot x-s(x)+q(x)\cdot (y-x)+\alpha s(x) \\
&\geq &q(z)\cdot x-s(z)+q(x)\cdot (y-x)+\alpha s(x),
\end{eqnarray*}%
where the second inequality follows because $(q(x),(1-\alpha )s(x))\in W,$
and the last inequality follows because $(q,s)$ is IC. Rearranging gives%
\begin{equation*}
\alpha (s(z)-s(x))\geq (q(x)-q(z))\cdot (y-x)-\varepsilon .
\end{equation*}%
Now $|(q(x)-q(z))\cdot (y-x)|\leq \rho $ (because $q(x),q(z)\in \lbrack
0,1]^{k}$ and $||y-x||_{1}\leq \rho ),$ and so%
\begin{equation*}
(1-\alpha )s(z)\geq (1-\alpha )s(x)-\frac{1-\alpha }{\alpha }(\rho
+\varepsilon ).
\end{equation*}%
Recalling that $\tilde{s}(y)\geq (1-\alpha )s(z)-\varepsilon $ yields%
\begin{equation*}
\tilde{s}(y)\geq (1-\alpha )s(x)-\frac{1-\alpha }{\alpha }(\rho +\varepsilon
)-\varepsilon ;
\end{equation*}%
as $\varepsilon >0$ was arbitrary, we have%
\begin{equation*}
\tilde{s}(y)\geq (1-\alpha )s(x)-\frac{1-\alpha }{\alpha }\rho ,
\end{equation*}%
and so, using $s(x)\leq ||x||_{1}\leq M,$%
\begin{equation}
\tilde{s}(y)\geq s(x)-\beta ,  \label{eq:s'}
\end{equation}%
where $\beta :=\alpha M+(1-\alpha )\rho /\alpha .$

For each $t>0$ put $A(t):=\{x\in D_{M}:s(x)\geq t\}$ and $\tilde{A}%
(t):=\{x\in D_{M}:\tilde{s}(x)\geq t\}.$ Inequality (\ref{eq:s'}), which
applies whenever $||y-x||_{1}\leq \rho ,$ implies that $\tilde{A}(t-\beta
)\supseteq B_{\rho }(A(t)),$ and so%
\begin{eqnarray*}
\text{\textsc{Rev}}(Y) &\geq &\mathbb{E}\left[ \tilde{s}(Y)\right]
=\int_{0}^{\infty }\mathbb{P}\left[ Y\in \tilde{A}(t)\right] \mathrm{d}t\geq
\int_{\beta }^{M}\mathbb{P}\left[ Y\in \tilde{A}(t-\beta )\right] \mathrm{d}t
\\
&\geq &\int_{\beta }^{M}\mathbb{P}\left[ Y\in B_{\rho }(A(t))\right] \mathrm{%
d}t\geq \int_{\beta }^{M}\mathbb{P}\left[ X\in A(t)\right] \mathrm{d}%
t-(M-\beta )\rho \\
&\geq &\int_{0}^{M}\mathbb{P}\left[ X\in A(t)\right] \mathrm{d}t-\beta
-(M-\beta )\rho =\mathbb{E}\left[ s(X)\right] -\beta ^{\prime },
\end{eqnarray*}%
where $\beta ^{\prime }:=\beta +(M-\beta )\rho $ (for the fourth inequality
we have used (\ref{eq:prohorov}) and $\beta \leq M,$ which follows from $%
\rho \leq \alpha $ and $M\geq 1$).

Taking $\alpha =\sqrt{\rho }$ gives $\beta \leq (M+1)\sqrt{\rho }$ and $%
\beta ^{\prime }\leq \beta +M\rho \leq (2M+1)\sqrt{\rho }.$ The displayed
inequality above holds for every $\mu $ and every $\rho >\ $\textrm{Dist}$%
(X,Y),$ and so \textsc{Rev}$(Y)\geq \ $\textsc{Rev}$(X)-(2M+1)\sqrt{\mathrm{%
Dist}(X,Y)}.$ Interchanging $X$ and $Y$ completes the proof.
\end{proof}

\bigskip

A sequence of random variables $(X^{n})_{n\geq 1}$ is \emph{uniformly
integrable} if for every\emph{\ }$\varepsilon >0$ there is a finite $M$ such
that $\mathbb{E}\left[ ||X^{n}||_{1}~\mathbf{1}_{||X^{n}||_{1}>M}\right]
\leq \varepsilon $ for all $n.$

\begin{theorem}
\label{th:continuous revenue-1}Let $X^{n}$ be a sequence of $k$-good random
valuations that converges in distribution to the $k$-good random valuation $%
X.$ Then 
\begin{equation*}
\liminf%
_{n\rightarrow \infty }\text{\textsc{Rev}}(X^{n})\geq \text{\textsc{Rev}}(X).
\end{equation*}%
Moreover, if the sequence $X^{n}$ is uniformly integrable, then 
\begin{equation*}
\lim_{n\rightarrow \infty }\text{\textsc{Rev}}(X^{n})=\text{\textsc{Rev}}%
(X)<\infty .
\end{equation*}
\end{theorem}

\begin{proof}
For every $M>0$ and every $k$-good valuation $X,$ denote $X_{(M)}:=X\mathbf{1%
}_{||X||_{1}\leq M}.$ Any IR, IC, and NPT mechanism $\mu =(q,s)$ satisfies $%
s\geq 0$ and $s(0)=0,$ and so $\mathbb{E}\left[ s(X_{(M)})\right] =\mathbb{E}%
\left[ s(X)\mathbf{1}_{||X||_{1}\leq M}\right] $ monotonically increases to $%
\mathbb{E}\left[ s(X)\right] $ as $M$ increases to infinity. Therefore 
\textsc{Rev}$(X_{(M)})$ monotonically increases to \textsc{Rev}$(X).$

If $X^{n}\overset{\mathcal{D}}{\rightarrow }X$ then $X_{(M)}^{n}\overset{%
\mathcal{D}}{\rightarrow }X_{(M)}$ for almost every $M>0$---specifically,
for those $M$ where\footnote{%
These are the points of continuity of the cumulative distribution function
of $||X||_{1};$ see Corollary 1 to Theorem 5.1 in Billingsley (1968).} $%
\mathbb{P}\left[ ||X||_{1}=M\right] =0$---and so $\lim_{n}$\textsc{Rev}$%
(X_{(M)}^{n})=\ $\textsc{Rev}$(X_{(M)})$ by Proposition \ref{p:dist}. Now 
\textsc{Rev}$(X^{n})\geq \ $\textsc{Rev}$(X_{(M)}^{n}),$ and
hence\linebreak\ $%
\liminf%
_{n}$\textsc{Rev}$(X^{n})\geq \ $\textsc{Rev}$(X_{(M)})$ for almost every $%
M. $ Letting $M\rightarrow \infty $ proves the first part of the theorem.

If in addition the sequence $X_{n}$ is uniformly integrable, then for every $%
\varepsilon >0$ there is $M>0$ with $\mathbb{P}\left[ ||X||_{1}=M\right] =0$
that is large enough so that $\mathbb{E}\left[ ||X^{n}||_{1}~\mathbf{1}%
_{||X^{n}||_{1}>M}\right] \leq \varepsilon $ for all $n.$ Since, as seen
above, $0\leq s(x)\leq ||x||_{1}$ for every IR and NPT mechanism $\mu =(q,s)$%
, it follows that $\mathbb{E}\left[ s(X^{n})\mathbf{1}_{||X^{n}||_{1}>M}%
\right] \leq \varepsilon $ for all $n$, and thus \textsc{Rev}$%
(X_{(M)}^{n})\geq \ $\textsc{Rev}$(X^{n})-\varepsilon $ for all $n$ (this
also shows that the revenues are all finite, as they are bounded by $%
M+\varepsilon ).$ Therefore \textsc{Rev}$(X)\geq \ $\textsc{Rev}$%
(X_{(M)})=\lim_{n}$\textsc{Rev}$(X_{(M)}^{n})\geq 
\limsup%
_{n}$\textsc{Rev}$(X^{n})-\varepsilon ,$ which, together with the first part
of the theorem, proves the second part.
\end{proof}

\subsection{Continuous Valuations\label{sus:wlog-continuity}}

It is often convenient---as in the present paper---to restrict attention to
random valuations whose distributions admit a density function (i.e., their
cumulative distribution functions are absolutely continuous; we refer to
these for now as \textquotedblleft continuous"). We now show that for the
results in the present paper one may restrict attention to continuous random
valuations. Indeed, assume that we have already proved a result of the form 
\textsc{Rev}$(X)\leq \theta ~$\textsc{SRev}$(X)$ for all such continuous $X,$
and let $X$ be a valuation that is not necessarily continuous (and so it may
have atoms and even finite support). First, because, as we have seen in the
proof of Theorem \ref{th:continuous revenue-1} above, \textsc{Rev}$%
(X_{(M)})\rightarrow _{M\rightarrow \infty }$\textsc{Rev}$(X)$ and \textsc{%
SRev}$(X_{(M)})\rightarrow _{M\rightarrow \infty }$\textsc{SRev}$(X)$, it
suffices to prove the result for random valuations $X$ with bounded values,
say $||X||_{1}\leq M.$ Second, let $U$ be independent of $X$ and distributed
uniformly on $[0,1]^{k},$ and for every $n$ define $X^{n}:=X+(1/n)U.$ Then,
clearly, the valuations $X^{n}$ are continuous, $X^{n}\overset{\mathcal{D}}{%
\rightarrow }X,$ and the sequence $X^{n}$ is bounded ($||X^{n}||_{1}\leq
||X||_{1}+(1/n)k\leq M+k)$; therefore \textsc{Rev}$(X^{n})\rightarrow
_{n\rightarrow \infty }$\textsc{Rev}$(X)$ and \textsc{SRev}$%
(X^{n})\rightarrow _{n\rightarrow \infty }$\textsc{SRev}$(X)$ (apply the
second part of Theorem \ref{th:continuous revenue-1} to the sequences $X^{n}%
\overset{\mathcal{D}}{\rightarrow }X$ and $X_{i}^{n}\overset{\mathcal{D}}{%
\rightarrow }X_{i}$ for all goods $i).$ Thus \textsc{Rev}$(X)\leq \theta ~$%
\textsc{SRev}$(X)$ holds for every bounded $X,$ and so for every $X.$

\section{Appendix: Nonsymmetric Diagonals\label{s:nonsymm}}

In this appendix we illustrate how the use of nonsymmetric diagonals alone
may strictly improve the $50\%$ bound of Hart and Nisan (2017), and, in some
cases, also the $62\%$ bound of our Theorem \ref{th:62.2}.\footnote{%
To get more than $50\%$ the single-good revenues need just to be different,
and to get more than $62\%$ they need to be significantly so (with the
revenue of one good about $9$ times higher than the revenue of the other).}
However, this improvement is \emph{not uniform}, in the sense that it does
not yield a better constant than\footnote{%
Yet another non-uniform bound may be obtained by optimizing $\lambda $ in
the first line of (\ref{eq:R1R2}): 
\begin{equation*}
\text{\textsc{Rev}}(X_{1},X_{2})\leq R_{1}+R_{2}+\min \left\{ \frac{2}{\sqrt{%
e}}\sqrt{R_{1}R_{2}},\frac{2}{e}\sqrt{R_{1}R_{2}}+\left( 1-\frac{1}{e}%
\right) \min \{R_{1},R_{2}\}\right\} ,
\end{equation*}%
where $R_{i}=$\textsc{Rev}$(X_{i})$ for $i=1,2.$} $50\%.$

\begin{proposition}
\label{p:sqrt(r1*r2)}Let $X=(X_{1},X_{2})$ be a two-good random valuation
with independent goods. Then%
\begin{equation*}
\text{\textsc{Rev}}(X_{1},X_{2})\leq \left( \sqrt{\text{\textsc{Rev}}(X_{1})}%
+\sqrt{\text{\textsc{Rev}}(X_{2})}\right) ^{2}.
\end{equation*}
\end{proposition}

\noindent \textbf{Remark. }When \textsc{Rev}$(X_{1})\neq \ $\textsc{Rev}$%
(X_{2})$ the right-hand side is strictly less than $2($\textsc{Rev}$%
(X_{1})+\ $\textsc{Rev}$(X_{2})),$ the bound of Theorem A of Hart and Nisan
(2017) (when \textsc{Rev}$(X_{1})=\ $\textsc{Rev}$(X_{2})$ the two bounds
are the same).

\bigskip

\begin{proof}
We follow the proof of Theorem A in Hart and Nisan (2017), but we now split
the computation along the diagonal $Y=\lambda Z$ for some $\lambda >0$
(instead of splitting along $Y=Z).$ The arguments in the proof there carry
over, and, for each fixed value $z$ of $Z,$ we now have%
\begin{eqnarray*}
\mathbb{E}\left[ s(Y,Z)\mathbf{1}_{Y\geq \lambda Z}|Z=z\right] &\leq &\text{%
\textsc{Rev}}(Y)+z\mathbb{P}\left[ Y\geq \lambda z\right] \\
&=&\text{\textsc{Rev}}(Y)+\frac{1}{\lambda }(\lambda z)\mathbb{P}\left[
Y\geq \lambda z\right] \\
&\leq &\text{\textsc{Rev}}(Y)+\frac{1}{\lambda }\text{\textsc{Rev}}%
(Y)=\left( 1+\frac{1}{\lambda }\right) \text{\textsc{Rev}}(Y).
\end{eqnarray*}%
Similarly, for each fixed value $y$ of $Y$, 
\begin{eqnarray*}
\mathbb{E}\left[ s(Y,Z)\mathbf{1}_{Z\geq (1/\lambda )Y}|Y=y\right] &\leq &%
\text{\textsc{Rev}}(Z)+y\mathbb{P}\left[ Z\geq \frac{y}{\lambda }\right] \\
&=&\text{\textsc{Rev}}(Z)+\lambda \frac{y}{\lambda }\mathbb{P}\left[ Z\geq 
\frac{y}{\lambda }\right] \\
&\leq &\text{\textsc{Rev}}(Z)+\lambda \text{\textsc{Rev}}(Z)=(1+\lambda )%
\text{\textsc{Rev}}(Z).
\end{eqnarray*}%
Taking expectation over the values of $Z$ and $Y,$ adding the two
inequalities, and then minimizing the resulting expression over $\lambda $
(by taking $\lambda =\sqrt{\text{\textsc{Rev}}(Y)/\text{\textsc{Rev}}(Z)}~$)
yields the result.
\end{proof}

\bigskip

\noindent \textbf{Remark. }A better bound than the one of Proposition \ref%
{p:sqrt(r1*r2)}, albeit also non-uniform, has been obtained by Kupfer (2017).


\begin{thebibliography}{99}
\bibitem{} Babaioff, M., N. Immorlica, B. Lucier, and S. M. Weinberg (2014),
\textquotedblleft A Simple and Approximately Optimal Mechanism for
an~Additive Buyer,\textquotedblright\ \emph{FOCS 2014: Proceedings of the
55th Annual Symposium on Foundations of Computer Science}, 21--30.

\bibitem{} Babaioff, M., N. Nisan, and A. Rubinstein (2018),
\textquotedblleft Optimal Deterministic Mechanisms for an Additive Buyer," 
\emph{EC 2018: Proceedings of the 19th ACM Conference on Economics and
Computation}, 429.

\bibitem{} Billingsley, P. (1968), \emph{Convergence of Probability Measures}%
, Wiley.

\bibitem{} Briest, P., S. Chawla, R. Kleinberg, and M. Weinberg (2015),
\textquotedblleft Pricing Randomized Allocations,\textquotedblright\ \emph{%
Journal of Economic Theory} 156, 144--174.

\bibitem{} Chawla, S., D. L. Malec, and B. Sivan (2010), \textquotedblleft
The Power of Randomness in Bayesian Optimal Mechanism
Design,\textquotedblright\ \emph{EC 2010: Proceedings of the 11th ACM
Conference on Electronic Commerce}, 149--158.

\bibitem{} Daskalakis, C., A. Deckelbaum, and C. Tzamos (2017),
\textquotedblleft Strong Duality for a Multiple-Good Monopolist," \emph{%
Econometrica} 85, 735--767.

\bibitem{} Hart, S. and N. Nisan (2013), \textquotedblleft The Menu-Size
Complexity of Auctions,\textquotedblright\ The Hebrew University of
Jerusalem, The Hebrew University of Jerusalem, Center for Rationality
DP-637; \emph{arXiv} 1304.6116; \emph{EC 2013: Proceedings of the 14th ACM
Conference on Electronic Commerce,} 565--566.

\bibitem{} Hart, S. and N. Nisan (2017), \textquotedblleft Approximate
Revenue Maximization with Multiple Items,\textquotedblright\ \emph{Journal
of Economic Theory} 172, 313--347; \emph{EC 2012: Proceedings of the 13th
ACM Conference on Electronic Commerce,} 656.

\bibitem{} Hart, S. and P. J. Reny (2015), \textquotedblleft Maximal Revenue
with Multiple Goods: Nonmonotonicity and Other
Observations,\textquotedblright\ \emph{Theoretical Economics} 10, 893--922.

\bibitem{} Kupfer, R. (2017), \textquotedblleft A Note on Approximate
Revenue Maximization with Two Items," \emph{arXiv} 1712.03518.

\bibitem{} Monteiro, P. K. (2015), \textquotedblleft A Note on the
Continuity of the Optimal Auction," \emph{Economics Letters} 137, 127--130.

\bibitem{} Myerson, R. B. (1981), \textquotedblleft Optimal Auction
Design,\textquotedblright\ \emph{Mathematics of Operations Research} 6,
58--73.

\bibitem{} Rockafellar, R. T. (1970), \emph{Convex Analysis}, Princeton
University Press.
\end{thebibliography}
\end{document}